\documentclass[conference]{IEEEtran}

\IEEEoverridecommandlockouts

\usepackage[cmex10]{amsmath}
\usepackage{amssymb,amsthm}

\newtheorem{thm}{Theorem}
\newtheorem{lem}[thm]{Lemma}

\newtheorem{cor}[thm]{Corollary}
\newtheorem{defini}[thm]{Definition}

\theoremstyle{definition}
\newtheorem{ex}[thm]{Example}
\newtheorem{rema}[thm]{Remark}

\newcommand{\field}[1]{\mathbb{#1}}
\newcommand{\F}{\field{F}}

\newcommand{\cM}{{\mathcal M}}

\newcommand{\PG}{\mathcal{P}_{q}(n)}

\newcommand{\G}{\mathcal{G}_{q}(k,n)}
\newcommand{\Uvs}{\mathcal{U}}
\newcommand{\Vvs}{\mathcal{V}}
\newcommand{\Cvs}{\mathcal{C}}
\newcommand{\C}{\mathcal{C}}

\newcommand{\rs}{\mathrm{rs}}

\newcommand{\enc}{\mathrm{enc}}
\newcommand{\GL}{\mathrm{GL}}
\newcommand{\stab}{\mathrm{stab}}
\newcommand{\ord}{\mathrm{ord}}

\begin{document}

\title{Message Encoding for Spread and Orbit Codes}

\author{\IEEEauthorblockN{Anna-Lena Trautmann}
\IEEEauthorblockA{Department of Electrical and Computer Systems Engineering\\
Monash University, Clayton, Australia}
\IEEEauthorblockA{Department of Electrical and Electronic Engineering\\
University of Melbourne, Australia\\
Email: anna-lena.trautmann@unimelb.edu.au}
\thanks{The author was supported by Swiss National Science Foundation Fellowship No.\ 147304.}
}

\maketitle

\begin{abstract}
Spread codes and orbit codes are special families of constant dimension subspace codes. These codes have been well-studied for their error correction capability and transmission rate, but the question of how to encode messages has not been investigated. In this work we show how the message space can be chosen for a given code and how message en- and decoding can be done.
\end{abstract}

\section{Introduction}

\emph{Subspace codes} are defined to be sets of subspaces of some given ambient space $\F_q^n$ of dimension $n$ over the finite field with $q$ elements. When we talk about \emph{constant dimension codes}, we restrict ourselves to subspace codes, whose codewords all have the same constant dimension.
Subspace codes in general, and constant dimension codes in particular, have received much attention since it was shown in \cite{ko08} how these codes can be used for random network coding.

In that same paper \cite{ko08} a class of Reed-Solomon-like codes is proposed, which was shown to be equivalent to the lifting of maximum rank distance codes \cite{si08j}. For theses codes one can easily find a suitable message space (or message set) $\cM$ and an encoding map, that maps $\cM$ injectively to the subspace code.

During the last years other constructions of subspace codes were developed, e.g.\ in \cite{bo09,et09,et12,et11,ga11,ga10,ko08p,ma08p,sk10,tr10}. Some of these constructions have the mere purpose of giving an improved transmission rate (i.e.\ larger cardinality of the code for the same parameters), while others also have some structure that can be used e.g.\ for decoding. The problem of message encoding has been addressed in almost none of these papers and is hence an open question for most of these codes. We want to study this problem for two classes of subspace codes, namely spread codes and orbit codes.

The paper is organized as follows: In the following section we will give some preliminaries, among others the spread code and orbit code construction. In Section \ref{sec:spread} we investigate a	natural message space and encoding map for Desarguesian spread codes, which we then extend to an encoding map on a set of integer numbers. In Section \ref{sec:orbit} we do the same for orbit codes. In Section \ref{sec:hybrid} we propose a hybrid encoding method, combining two encoding and decoding algorithms for spread codes. We conclude this work in Section \ref{conclusion}.


\section{Preliminaries}\label{sec:preliminaries}

We denote the finite field with $q$ elements by $\F_q$. The set of all subspaces of $\F_q^n$ is denoted by $\PG$ and the set of all subspaces of $\F_q^n$ of dimension $k$, called the \emph{Grassmannian}, is denoted by $\G$. We represent a vector space $\Uvs \in \G$ by a matrix $U\in \F_q^{k\times n}$ such that the row space of $U$, denoted by $\rs(U)$, is equal to $\Uvs$. 
A subspace code is simply a subset of $\PG$ and a constant dimension code is a subset of $\G$. A metric on $\PG$ is given by the subspace distance (\cite{ko08})
$$ d_S(\Uvs, \Vvs) := \dim(\Uvs) + \dim(\Vvs)- 2\dim(\Uvs\cap \Vvs) $$
for any $\Uvs, \Vvs \in \PG$. The minimum distance $d_S(\Cvs)$ of a subspace code $\C \subseteq \PG$ is the minimum of all the pairwise distances of the codewords. Since the dual of a subspace code $\C$ has the same minimum distance as $\C$ (see e.g.\ \cite{ko08}), it is customary to restrict oneself to $k\leq n/2$, which we will assume throughout the paper.

A \emph{spread code} \cite{ma08p} in $\G$ is defined as a set of elements of $\G$ that pairwise intersect only trivially and cover the whole space $\F_q^n$. They exist if and only if $k| n$, have minimum distance $2k$ and cardinality $(q^n-1)/(q^k-1)$. For more information on different constructions and decoding algorithms of spread codes see \cite{go12,ma08p,ma11j,tr13phd}. We will use the following construction, which gives rise to a \emph{Desarguesian spread code} in $\G$ (\cite{tr13phd}):
\begin{enumerate}
\item Let $m:=n/k$ and consider $\mathcal G_{q^k} (1,m)$, which has $q^{k(m-1)}+q^{k(m-2)} + q^{k(m-3)} + \dots + 1 = (q^n-1)/(q^k-1)$ elements. Trivially, all these lines intersect only trivially.
\item Let $P$ be the companion matrix of an irreducible polynomial over $\F_q$ of degree $k$. Then it holds that $\F_{q^k} \cong \F_q[P]$ and we can use this isomorphism in any element of $\mathcal G_{q^k} (1,m)$ (i.e.\ we replace any coordinate with the respective matrix) to receive a spread code in $\G$.
\end{enumerate}

\begin{ex}\label{ex1}
Let $\alpha$ be a root of $x^2+x+1$, i.e.\ a  primitive element of $\F_{2^2}\cong \F_2[\alpha]$. The respective companion matrix is 
$$P=\left(\begin{array}{cc} 0&1\\1&1 \end{array}\right) .$$
Then $\mathcal G_{2^2} (1,2) = \{ \rs (1 , 0), \rs(1 , \alpha), \rs (1 , \alpha^2), \rs(1 , 1), \rs(0 , 1) \}$ and substituting all elements of $\F_{2^2}\cong\F_2[\alpha]$ with its corresponding element from $\F_2[P]$ gives a spread in $\mathcal G_2(2,4)$.
\end{ex}

\emph{Orbit codes} \cite{tr10p} in $\G$ are defined to be orbits of a subgroup of the general linear group $\GL_n$ of order $n$ over $\F_q$. They can be seen as the analogs of linear codes in classical block coding and their structure can be used for an easy computation of the minimum distance of a code and for decoding algorithms (e.g.\ one can define coset leader decoding for them). For more information on orbit codes the interested reader is referred to \cite{ma11p,ro12j,tr13phd,tr11a}. One can also use the orbit code construction to construct spread codes. Note that this construction of spread codes is not equivalent to the Desarguesian construction from before.
\begin{ex}\label{ex2}
The following orbit code is also a spread code in $\mathcal G_2(2,4)$:
$$\rs\left(\begin{array}{cccc} 1&0&0&0\\0&1&1&0 \end{array}\right) \left\langle\left(\begin{array}{cccc} 0&1&0&0 \\ 0&0&1&0 \\ 0&0&0&1 \\ 1&1&0&0 \end{array}\right) \right \rangle$$
\end{ex}

In general, for any given code $\C$ in some space $X$ and some message space $\cM$, the corresponding encoding map
$$\enc : \cM \longrightarrow X$$
is an injective map, mapping any message to a codeword. I.e.\ $\enc(\cM)=\C$. Mostly in the information theory literature $\cM=\{0,\dots,j-1\}$ for some integer $j$. For classical linear block codes the usual message space is $\cM= \F_q^k$ for some integer $k$. If $q=p^r$ for some prime number $p$, then $\F_q^k \cong \F_p^{rk}$ and the \emph{$p$-adic expansion}
\begin{align*}
\phi : \quad \F_p^{rk} &\longrightarrow \{0,\dots,rk-1\}\\
(u_0,\dots, u_{rk-1}) &\longmapsto x=\sum_{i=0}^{rk-1} u_i p^i
\end{align*}
is a bijection. Moreover, $\phi$ and $\phi^{-1}$ can be computed very efficiently (for the inverse one recursively computes $u_{i+1} \equiv (x - u_i)/p \mod p$ with the initial congruence $u_0 \equiv x \mod p$).

In the subspace coding case it is not obvious what $\cM$ would be and how message encoding or decoding can be done. An elegant solution is given for the Reed-Solomon-like codes in \cite{ko08}. For such a code $\C \subseteq \G$ the message space is
$$\cM=\F_{q^{n-k}}^{k-\frac{d_S(\C)}{2}+1} ,$$
which is isomorphic (as a vector space) to $\F_{q}^{(n-k)(k-d_S(\C)/2+1)}$, and the encoding map is given by
\begin{align*}
 \enc : \quad \F_{q^{n-k}}^{k-\frac{d_S(\C)}{2}+1}  &\longrightarrow  \G \\
 (v_0,\dots, v_{k-\frac{d_S(\C)}{2}}) &\longmapsto \langle  (\beta_j, \sum_{i=0}^{k-\frac{d_S(\C)}{2}} v_i \beta_j^i)  \mid j=1,\dots, k\rangle
\end{align*}
where $\beta_1,\dots, \beta_k$ is a basis of $\F_{q^k}$ over $\F_q$ and we use $\F_q^n\cong \F_{q^n}$ on the right side. Via interpolation this map is invertible and the inverse is computable in polynomial time. Hence, one gets a feasible message decoding map as well.

In the following sections we want to investigate if one can find message encoding maps from a set of integers to orbit and spread codes, whose inverse is efficiently computable, as well.


\section{Message Encoding for Desarguesian Spread Codes}\label{sec:spread}

We call a spread in $\G$ Desarguesian if it is isomorphic to $\mathcal G_{q^k} (1,m)$ (where $m=n/k$). For simplicity though, we will work only with codes arising from the construction as described in the previous section. Analog results for the equivalent codes can then easily be derived.

Because of the isomorphic description of the code as all elements of $\mathcal G_{q^k} (1,m)$, the easiest choice of message space is exactly $\cM = \mathcal G_{q^k} (1,m)$ and the encoding map is the second point of the construction in Section \ref{sec:preliminaries}. Let $\alpha$ be a primitive element of $\F_{q^k}$, $p_{\alpha}(x) \in \F_q[x]$ its minimal polynomial and $P_{\alpha} \in \GL_k$ the corresponding companion matrix. Then $\F_{q^k} \cong \F_q[\alpha]$ and any element in $\F_{q^k}$ can be expressed as a polynomial in $\alpha$ of degree less than $k$, and one can define the following encoding map:
\begin{align*}
 \enc_1 : \quad \mathcal G_{q^k} (1,m)  &\longrightarrow  \G \\
 \rs(\sum_{i=0}^{k-1} u_{1i} \alpha^i , \dots, \sum_{i=0}^{k-1} u_{mi} \alpha^i) &\longmapsto
 \rs(\sum_{i=0}^{k-1} u_{1i} P^i , \dots, \sum_{i=0}^{k-1} u_{mi} P^i)    .
\end{align*}
This map is well defined, since all non-zero elements of $\F_q[P]$ have full rank and hence the right side is always and element of $\G$. Note that the left side is represented by a basis vector over $\F_{q^k}$, whereas the right side is represented by a matrix in $\F_q^{k\times n}$, whose row space is the corresponding codeword.

\begin{thm}
The map $\enc_1$ is injective.
\end{thm}
\begin{proof}
This follows from the isomorphism $\F_q[\alpha] \cong \F_{q}[P]$, since $f(\alpha) = g(\alpha)$ if and only if $f(P) = g(P)$ for any $f(x),g(x) \in \F_q[x]$.
\end{proof}

Thus, one can derive an inverse map, called the decoding map. In this case the decoding map is again very simple, and since none of the codewords intersect in a non-zero element, it is enough to consider only one non-zero vector $v\in \F_q^n$ of the codeword to recover the message.
For this we translate that vector $v$ into a vector over $\F_{q^k}$, i.e.\ we partition $v$ into blocks of length $k$ and represent these blocks in their extension field representation ($\F_q^k \cong \F_{q^k} \cong \F_{q}[\alpha]$). This is then a basis of the corresponding message in $\mathcal G_{q^k} (1,m)$.

If one wants to have a unique description of the messages, one can choose the normalized basis vector, i.e.\ the one element of the one-dimensional subspace whose first non-zero entry is equal to one. In the message decoding process, one needs to add an additional step then, that divides all elements of the vector  in $\F_{q^k}^m$ by the first non-zero entry of that new vector.

The reader familiar with projective spaces will notice that $\mathcal G_{q^k} (1,m)$ corresponds exactly to the projective space over $\F_{q^k}$ of dimension $m-1$. The usage of a normalized representative of points in that space is a common concept there.

\begin{thm}\label{thm1}
For a code $\C \subseteq \G$ the decoding map $\enc_1^{-1} :\C \rightarrow \mathcal G_{q^k} (1,m)$ can be computed with a complexity of order $\mathcal{O}_q (kn)$.
\end{thm}
\begin{proof}
Choose one vector $v\in \F_q^n$ of the given codeword and represent it as an element of $\F_{q^k}^m$. 
For the normalization, one needs at most $m = n/k$ divisions over $\F_{q^k}$. Each such division can be done with $\mathcal{O}_q (k^2)$ operations. 
\end{proof}

Note, that in the spread decoding algorithm of $\cite{ma11j}$  one gets the normalized representation of the message along the way in the algorithm and the additional step of message decoding is not necessary.

In the following we will show how one can also encode the message set $\cM = \{1,\dots, (q^n-1)/(q^k-1)\}$ by concatenating $\enc_1$ with yet another map:
\begin{align*}
 f : \{1,\dots, (q^n-1)/(q^k-1)\} \longrightarrow \mathcal G_{q^k} (1,m)   \\
  i \longmapsto  \langle (\underbrace{0, \dots ,0 }_{\epsilon(i)} ,1,  \phi_i^{-1}(i- \sum_{j=0}^{\epsilon(i) -1} q^{jk}) )\rangle   .
\end{align*}
where $\epsilon(i) := m-\min\{y \mid \sum_{j=0}^{y-1} q^{jk} \geq i\}$ and $\phi_i : \F_{q^k}^{m-\epsilon(i) -1} \rightarrow\{1,\dots, q^{k(m-\epsilon(i) -1)}\} $ is the $p$-adic expansion, as explained in Section \ref{sec:preliminaries}.

\begin{thm}
The map $f$ is bijective and hence
$$\enc_2 := \enc_1 \circ f$$
is an injective map from $\{1,\dots, (q^n-1)/(q^k-1)\}$ to $\G$. 
\end{thm}
\begin{proof}
We show that $f$ is injective, then by the equal cardinalities of domain and codomain it is automatically bijective. It holds that $1$ is mapped to $\langle(0,\dots,0,1)\rangle$, $\{2,\dots,q^k+1\}$ is mapped to $\langle(0,\dots,0,1, \F_{q^k})\rangle$, $\{q^k+2,\dots,q^{2k}+q^k+1\}$ is mapped to $\langle(0,\dots,0,1, \F_{q^k}^2)\rangle$, etc. Since $\phi_i$ is bijective, the statement follows.
\end{proof}

As before, one can easily find the inverse map of $\enc_2$ and get a message decoding map for the integer message set as well.

\begin{thm}
The maps $\enc_2$ and $\enc_2^{-1}$ are computable with a computational complexity of order at most $\mathcal O_q(kn)$.
\end{thm}
\begin{proof}
Since $\epsilon(i)$ only takes $m-1$ values, one can store these in a look-up table and use an ordered search to find the right value. (But also computing $\epsilon(i)$ without a table can be done efficiently.)
Since $\phi_i$ and $\phi_i^{-1}$ are efficiently computable, the overall complexity of the inverse map is dominated by the normalization (see Theorem \ref{thm1}). Since the complexity of $\enc_2$ is lower than the one of $\enc_2^{-1}$, the statement follows.
\end{proof}

Note that due to simplicity we chose $\cM = \{1,\dots, (q^n-1)/(q^k-1)\}$, but clearly one can change $f$ and thus $\enc_2$ to encode the message set $\{0,\dots, (q^n-1)/(q^k-1)-1\}$.


\section{Message Encoding for Cyclic Orbit Codes}\label{sec:orbit}

Recall that an orbit code $\C\subseteq \G$ is defined as the orbit of a given $\Uvs\in \G$ under the action of a subgroup $G$ of $\GL_n$. In general it holds that $|\Cvs | \leq |G|$, i.e.\ some elements of $G$ might generate the same codewords. Denote by 
$$\stab_{\GL_n}(\Uvs) := \{A\in \GL_n \mid \Uvs A = \Uvs\}$$ 
the stabilizer of $\Uvs$ in $\GL_n$, and by $G/\stab_{\GL_n}(\Uvs)$ the set of all right cosets $\stab_{\GL_n}(\Uvs) A$ for $A \in \GL_n$.  
Then the encoding map can be defined as
\begin{align*}
 \enc_3 : \quad G/\stab_{\GL_n}(\Uvs)  &\longrightarrow \G \\
[A] &\longmapsto \Uvs A  .
\end{align*}
where $[A]$ denotes the coset of $A$.

\begin{thm}
The map $\enc_3$ is injective.
\end{thm}
\begin{proof}
Let $A,B \in G$. Assume that  $\Uvs A = \Uvs B$, then 
$$AB^{-1} \in \stab_{\GL_n}(\Uvs)$$
and thus $A = AB^{-1}B \in \stab_{\GL_n}(\Uvs) B$. Hence, $A$ and $B$ are in the same right cosets of $\stab_{\GL_n}(\Uvs) $. 
\end{proof}

We now want to find an encoding map for orbit codes with respect to the integer numbers as messages. To do so we will restrict ourselves to cyclic orbit codes in this paper, since these have more useful structure. Moreover, cyclic orbit codes are also better understood from a construction and error decoding point of view.

Cyclic orbit codes are those codes that can be defined by the action of a cyclic subgroup $G$, i.e.\ $G=\langle P\rangle$ for some matrix $P\in \GL_n$. Then one clearly has a bijection from $\cM =\{0,\dots, \ord(P)-1 \}$ to $G$:
\begin{align*}
 g' : \quad  \{0,\dots, \ord(P)-1 \} &\longrightarrow G  \\
i &\longmapsto P^i .
\end{align*}
From group theory (see e.g.\ \cite{ke99}) one knows that $|G/\stab_{\GL_n}(\Uvs)|$ is a divisor of $|G| = \ord(P)$ and that if $\ord_{\Uvs}(P):= |G/\stab_{\GL_n}(\Uvs)| < |G|$, then $\Uvs P^i = \Uvs P^{i+ \ord_{\Uvs}(P)}$. Thus it follows:

\begin{lem}
The map
\begin{align*}
 g : \quad  \{0,\dots, \ord_{\Uvs}(P)-1 \} &\longrightarrow G/\stab_{\GL_n}(\Uvs)  \\
i &\longmapsto P^i .
\end{align*}
is a bijection for any $\Uvs \in \G$.
\end{lem}

\begin{cor}
The map $\enc_4 := \enc_3 \circ g$ is injective and hence an encoding map for the message set $\cM =\{0,\dots, \ord_{\Uvs}(P)-1 \}$.
\end{cor}

Note that $\enc_4$ can be computed very efficiently while its inverse is a discrete logarithm problem (DLP), which is in general a hard problem. There are many results on when the DLP is hard and when it is not; for a survey of various algorithms and their complexities see e.g.\ \cite{od85}. In the following we will investigate some of the easy cases, since these will be the one of interest from an application point of view.


\subsection{Primitive Cyclic Orbit Codes}

For this subsection let $\alpha$ be a primitive element of $\F_{q^n}$, $p_{\alpha}(x) \in \F_q[x]$ its minimal polynomial and $P_{\alpha}$ the corresponding companion matrix. Denote by $G = \langle P_{\alpha} \rangle $ the group generated by it. Because of the primitivity it holds that
$$ \ord(\alpha) = \ord (P_{\alpha}) = |G| = q^n-1 .$$
We call $\C = \Uvs G $ a primitive cyclic orbit code for any $\Uvs\in \G$. For more information on the cardinality and minimum distance of different primitive cyclic orbit codes the interested reader is referred to \cite{tr11a}, but we want to remark that for any valid set of parameters one can construct a spread code as a primitive cyclic orbit code. In this case one constructs $\Uvs$ in such a way that $q^k-1$ of its non-zero elements are in its own stabilizer $\stab_{\GL_n}(\Uvs)$ and hence $G/\stab_{\GL_n}(\Uvs) = (q^n-1)/(q^k-1)$.

Using the Pohlig-Hellman algorithm for DLP \cite[Sec. 3.6.3]{me97b}, one can compute a solution for the discrete logarithm with a computational complexity of order $\mathcal{O}_{q^n}(\sum_{i=1}^r e_i(\log_2 q^n+\sqrt{p_i})) \leq \mathcal{O}_{q}(n^2 \sum_{i=1}^r e_i(\log_2 q^n+ \sqrt{p_i}))$ where $\prod_{i=1}^r p_i^{e_i} $ is the prime factorization of $q^n-1$. 

For simplicity we will now concentrate on the case $q=2$. If $q=2$, the above complexity becomes
\[\mathcal{O}_{2}(n^2 \sum_{i=1}^r e_i(n+\sqrt{p_i})) .\]
Hence, if $2^n-1$ is $n^2$-smooth (i.e.\ if all prime factors of $2^n-1$ are less than or equal to $n^2$) and the largest $e_i$ is less than or equal to $k$, then the order of this complexity is upper bounded by $\mathcal{O}_2(n^3 k)$, which is reasonable. For this note e.g.\ that the complexities of the decoders in \cite{ko08,si08j} are at least cubic in $n$. The decoding complexities of the two error decoding algorithms for primitive cyclic orbit codes in \cite{tr11a} are of order $\mathcal{O}_2(4^{k}(n^2 + k^2 n))$ and $\mathcal{O}_2(nk(nk-k^2-n))$, respectively. Thus, in most cases, the message decoding would not drastically increase the overall complexity.

Table \ref{table1} shows values of $n$ for which $2^n-1$ is $n^2$-smooth. As one can see, also the largest exponent $e_i$ is small, hence the above statement holds for many values of $k$.

\begin{table}
\begin{center}
\begin{tabular}{|c|c|c|c|c|c|}
\hline
$n$ & 
max $p_i$ & max $e_i$ &  $ \max(e_i n, e_i p_i)$ & $n^2$ \\
\hline
$6$ 
& $7  $ & 3 & 18 & 36\\
$8$ 
& 17 & 1 & 17 & 64\\
$9$ 
& 73 &1& 73 & 81\\
$10$ 
& 31 &1& 31 & 100\\
$11$ 
& 89 &1 & 89 & 121\\
$12$ & 13 & 2& 24 &144 
\\
14 & 127& 1& 127 &196
\\
15& 151 & 1 & 151  &225
\\
18& 73 & 3 & 73 &324
\\
20& 41 & 2 & 41 & 400
\\
21& 337 & 2 & 337&441 
\\
24& 241 & 2 &  241 &576
\\
28& 127 & 1 & 127& 784
\\
30& 331 & 2 & 331 & 900
\\
36& 109 & 3 & 109 & 1296
\\
48 & 673 & 2 & 673 & 2304
\\
60 & 1321 & 2 & 1321 & 3600
\\
\hline
\end{tabular}
\caption{$n^2$-smooth $2^n-1 = \prod_{i=1}^r p_i^{e_i}$.}\label{table1}
\end{center}
\end{table}

Thus, we have shown that there exist parameters for which $\enc_4$ is a message encoding function for orbit codes, that has an efficient inverse map, i.e.\ an efficient corresponding decoder. For many parameters though, the procedures described in this section are not efficiently computable, which is why we derive other algorithms for the special class of orbit spread codes in the next section.


%

\section{A Hybrid En- and Decoder for Spread Codes}\label{sec:hybrid}

As mentioned in the previous section, orbit codes have useful structure, which can be exploited for error decoding. E.g.\ the coset leader decoding algorithm for irreducible cyclic orbit codes from \cite{tr11a} has a very low computational complexity. Spread codes are among the most interesting constant dimension codes because of their optimal tradeoff between error correction capability and transmission rate. As mentioned before, they can be constructed as primitive cyclic orbit codes, and we can hence use the coset leader decoder for them. On the other hand, we have an efficient message en- and decoder for Desarguesian spreads, as described in Section \ref{sec:spread}. In this section we want to combine the message en- and decoder for Desarguesian spread codes with the error correction en- and decoder for orbit codes, which we call a \emph{hybrid en- and decoder} for spread codes. 

For this assume that there exist a Desarguesian spread code $S_1 \in \G$ and a primitive cyclic orbit spread code $S_2\in \G$, such that $S_1 A = S_2$ (as sets of vector spaces) for some $A\in \GL_n$. Then we can define the following encoding map for the message space $\cM=\left\{1,\dots, (q^n-1)(q^k-1) \right\}$:
\begin{align*}
 \enc_5 : \quad \cM &\longrightarrow  \G \\
 i & \longmapsto \enc_2(i) A
 \end{align*}

\begin{thm}
The map $\enc_5$ is injective and both $\enc_5$ and $\enc_5^{-1}$ are computable with a computational complexity of order at most $\mathcal O_q(kn^2)$.
\end{thm}
\begin{proof}
The multiplication with $A$ can be done with the order of $\mathcal O_q(kn^2)$, which dominates the complexity order of $\enc_2$. The inverse $A^{-1}$ can be precomputed and stored and hence in the decoding map the multiplication with $A^{-1}$ has the same complexity, or only $\mathcal O_q(n^2)$, if we use only only one vector as representative of the whole vector space. The same computations can naturally also be done in the extension field representation, using $\F_q^n \cong \F_{q^n}$.
\end{proof}

Moreover, $\enc_5(\cM) = S_2$, i.e.\ we send codewords of $S_2$ over the channel and can use the corresponding error decoding algorithms for cyclic orbit codes, before we then apply $\enc_5^{-1}$ to recover the message. Note that this gives an efficient message en- and decoder for primitive cyclic orbit spread codes, independent of the discrete logarithm problem.

It remains to show that there are Desarguesian spread codes that are related to primitive cyclic orbit codes by a linear transformation.
In this case one also says that they are linearly isometric (see \cite{tr13phd,tr12}). 
It was shown in \cite{tr12} that not all spreads are linearly isometric, i.e.\ you cannot always find a linear map from one spread in $\G$ to another spread in $\G$. On the other hand, it was also shown that all Desarguesian spreads are linearly isometric. Hence, for our purposes, it remains to investigate when a primitive cyclic orbit spread code is a Desarguesian spread. This can be done by using the algorithm of \cite{thomas}, or by using the following results: Desarguesian spreads are always orbit codes \cite{tr10p}, and two orbit codes are linearly isometric if and only if their generating groups are conjugates \cite{ma11p}. Furthermore, if one of two conjugate groups is cyclic, also the other one is cyclic and there exist two respective generator matrices of the two groups, that are similar. This way, one can check if a given Desarguesian spread code is linearly isometric to a given cyclic orbit spread code. We will illustrate one such pair of codes in the following and concluding example.

\begin{ex}
Let $S_1$ be the spread constructed in Example \ref{ex1} and let $S_2$ be the orbit spread code constructed in Example \ref{ex2}, both subsets of $\mathcal G_2(2,4)$. 
Then
$$ A=\left(\begin{array}{cccc} 1&0&0&0 \\ 0&1&1&0 \\ 1&1&0&0 \\ 0&1&0&1 \end{array}\right) $$
is a linear transformation from $S_1$ to $S_2$. Let $\beta$ be a primitive element of $\F_{q^n}$. 
In the isomorphic extension field representation, $A$ maps the basis $\{1,\beta, \beta^2, \beta^3\}$ of $\F_{q^n}$ over $\F_q$ to the new basis $\{1, \beta+ \beta^2, 1+\beta, \beta+ \beta^3\}$. We can now use $S_1$ for message encoding, say we got the codeword $c \in S_1$, then we send the codeword $c A \in S_2$ over the channel. We can then do error correction decoding in the code $S_2$ with any orbit decoder (e.g.\ with coset leader decoding), say we get the codeword $c' \in S_2$, and transform it to $c' A^{-1} \in S_2$, from which we can then easily get the message as explained in Section \ref{sec:spread}.
\end{ex}


\section{Conclusion}\label{conclusion}

In this work we investigate how message encoding can be done for spread and orbit codes, two families of subspace codes that have been well studied for error correction in random network coding. 

We show that for Desarguesian spread codes one can find encoding maps such that the map itself and the inverse map are efficiently computable. We also show that for general cyclic orbit codes message decoding translates to a discrete logarithm problem, which is efficiently computable for some sets of parameters, but not in general. In the end we propose a hybrid en- and decoder for spread codes, such that one can use the orbit structure for error correction, but avoid the discrete logarithm problem in the message decoding part.

The results for orbit codes are shown for primitive cyclic orbit codes, but a generalization to arbitrary irreducible cyclic orbit codes is straight-forward. Furthermore, with some more effort one can then generalize these results to general cyclic orbit codes.

An open question for further research is if one can find general results on when cyclic orbit spread codes are Desarguesian and how to find the linear transformation from one spread into the other without the help of the algorithm of $\cite{thomas}$. Moreover, one can investigate if there are other codes where a hybrid en- and decoder can be helpful to combine efficient error correction decoders with efficient message decoders.


\bibliographystyle{plain}

\bibliography{network_coding_stuff}

\begin{thebibliography}{10}

\bibitem{bo09}
M.~Bossert and E.M. Gabidulin.
\newblock One family of algebraic codes for network coding.
\newblock In {\em Information Theory, 2009. ISIT 2009. IEEE International
  Symposium on}, pages 2863--2866, 2009.

\bibitem{et09}
T.~Etzion and N.~Silberstein.
\newblock Error-correcting codes in projective spaces via rank-metric codes and
  {F}errers diagrams.
\newblock {\em IEEE Transactions on Information Theory}, 55(7):2909--2919,
  March 2009.

\bibitem{et12}
T.~Etzion and N.~Silberstein.
\newblock Codes and designs related to lifted {MRD} codes.
\newblock {\em IEEE Transactions on Information Theory}, 59(2):1004 --1017,
  2013.

\bibitem{et11}
T.~Etzion and A.~Vardy.
\newblock Error-correcting codes in projective space.
\newblock {\em IEEE Transactions on Information Theory}, 57(2):1165--1173,
  2011.

\bibitem{thomas}
T.~Feulner.
\newblock Canonical forms and automorphisms in the projective space.
\newblock {\em preprint}, 2012.

\bibitem{ga11}
E.~M. Gabidulin and N.~I. Pilipchuk.
\newblock Multicomponent network coding.
\newblock In {\em Proceedings of the Seventh International Workshop on Coding
  and Cryptography (WCC) 2011}, pages 443--452, Paris, France, 2011.

\bibitem{ga10}
M.~Gadouleau and Z.~Yan.
\newblock Constant-rank codes and their connection to constant-dimension codes.
\newblock {\em IEEE Transactions on Information Theory}, 56(7):3207--3216,
  2010.

\bibitem{go12}
E.~Gorla, F.~Manganiello, and J.~Rosenthal.
\newblock An algebraic approach for decoding spread codes.
\newblock {\em Advances in Mathematics of Communications (AMC)}, 6(4):443 --
  466, 2012.

\bibitem{ke99}
A.~Kerber.
\newblock {\em Applied finite group actions}, volume~19 of {\em Algorithms and
  Combinatorics}.
\newblock Springer-Verlag, Berlin, second edition, 1999.

\bibitem{ko08p}
A.~Kohnert and S.~Kurz.
\newblock Construction of large constant dimension codes with a prescribed
  minimum distance.
\newblock In J.~Calmet, W.~Geiselmann, and J.~M{\"u}ller-Quade, editors, {\em
  MMICS}, volume 5393 of {\em Lecture Notes in Computer Science}, pages 31--42.
  Springer, 2008.

\bibitem{ko08}
R.~K{\"o}tter and F.~R. Kschischang.
\newblock Coding for errors and erasures in random network coding.
\newblock {\em IEEE Transactions on Information Theory}, 54(8):3579--3591,
  2008.

\bibitem{ma08p}
F.~Manganiello, E.~Gorla, and J.~Rosenthal.
\newblock Spread codes and spread decoding in network coding.
\newblock In {\em Proceedings of the 2008 IEEE International Symposium on
  Information Theory}, pages 851--855, Toronto, Canada, 2008.

\bibitem{ma11j}
F.~Manganiello and A.-L. Trautmann.
\newblock Spread decoding in extension fields.
\newblock {\em arXiv:1108.5881v1 [cs.IT]}, to appear in Finite Fields and
  Applications, 2011.

\bibitem{ma11p}
F.~Manganiello, A.-L. Trautmann, and J.~Rosenthal.
\newblock On conjugacy classes of subgroups of the general linear group and
  cyclic orbit codes.
\newblock In {\em Proceedings of the 2011 IEEE International Symposium on
  Information Theory}, pages 1916--1920, St. Petersburg, Russia, 2011.

\bibitem{me97b}
A.~J. Menezes, P.~C. van Oorschot, and S.~A. Vanstone.
\newblock {\em Handbook of applied cryptography}.
\newblock CRC Press Series on Discrete Mathematics and its Applications. CRC
  Press, Boca Raton, FL, 1997.
\newblock With a foreword by Ronald L. Rivest.

\bibitem{od85}
A.M. Odlyzko.
\newblock Discrete logarithms in finite fields and their cryptographic
  significance.
\newblock In Thomas Beth, Norbert Cot, and Ingemar Ingemarsson, editors, {\em
  Advances in Cryptology}, volume 209 of {\em Lecture Notes in Computer
  Science}, pages 224--314. Springer Berlin Heidelberg, 1985.

\bibitem{ro12j}
J.~Rosenthal and A.-L. Trautmann.
\newblock A complete characterization of irreducible cyclic orbit codes and
  their {P}l\"ucker embedding.
\newblock {\em Designs, Codes and Cryptography}, 66:275--289, 2013.

\bibitem{si08j}
D.~Silva, F.~R. Kschischang, and R.~K\"otter.
\newblock A rank-metric approach to error control in random network coding.
\newblock {\em IEEE Transactions on Information Theory}, 54(9):3951 --3967,
  2008.

\bibitem{sk10}
V.~Skachek.
\newblock Recursive code construction for random networks.
\newblock {\em Information Theory, IEEE Transactions on}, 56(3):1378--1382,
  2010.

\bibitem{tr13phd}
A.-L. Trautmann.
\newblock {\em Constructions, Decoding and Automorphisms of Subspace Codes}.
\newblock PhD thesis, University of Zurich, Switzerland, 2013.

\bibitem{tr12}
A.-L. Trautmann.
\newblock Isometry and automorphisms of constant dimension codes.
\newblock {\em Advances in Mathematics of Communications (AMC)}, 7(2):147--160,
  2013.

\bibitem{tr11a}
A.-L. Trautmann, F.~Manganiello, M.~Braun, and J.~Rosenthal.
\newblock Cyclic orbit codes.
\newblock {\em IEEE Transactions on Information Theory}, 59(11):7386--7404,
  2013.

\bibitem{tr10p}
A.-L. Trautmann, F.~Manganiello, and J.~Rosenthal.
\newblock Orbit codes - a new concept in the area of network coding.
\newblock In {\em IEEE Information Theory Workshop (ITW)}, pages 1--4, Dublin,
  Ireland, 2010.

\bibitem{tr10}
A.-L. Trautmann and J.~Rosenthal.
\newblock New improvements on the echelon-{F}errers construction.
\newblock In {\em Proceedings of the 19th International Symposium on
  Mathematical Theory of Networks and Systems -- MTNS}, pages 405--408,
  Budapest, Hungary, 2010.

\end{thebibliography}

\end{document}